\documentclass[copyright]{eptcs}
\providecommand{\event}{FOCLASA 2011}

\usepackage[numbers,sort,square,comma]{natbib}
\usepackage{amsmath,array,url}
\usepackage{eptcs_thm,foclasa2011}
\usepackage{graphicx,psfrag}

\def\figscale{.55}
\begin{document}


\title{A State-Based Characterisation of the Conflict Preorder}

\author{Simon Ware \qquad Robi Malik
        \institute{Department of Computer Science,
                   University of Waikato, Hamilton, New Zealand}
        \email{$\{$siw4,robi$\}$@waikato.ac.nz}}

\def\titlerunning{A State-Based Characterisation of the Conflict Preorder}
\def\authorrunning{Simon Ware and Robi Malik}

\maketitle

\begin{abstract}
  This paper proposes a way to effectively compare the potential of
  processes to cause \emph{conflict}. In discrete event systems theory, two
  concurrent systems are said to be in conflict if they can get trapped in
  a situation where they are both waiting or running endlessly, forever
  unable to complete their common task. The \emph{conflict preorder} is a
  process-algebraic pre-congruence that compares two processes based on
  their possible conflicts in combination with other processes. This paper
  improves on previous theoretical descriptions of the conflict preorder
  by introducing \emph{less conflicting pairs} as a concrete state-based
  characterisation. Based on this characterisation, an effective algorithm
  is presented to determine whether two processes are related according to
  the conflict preorder.
\end{abstract}


\section{Introduction}

A key question in process algebra is how processes can be composed and
compared~\cite{DeNHen:84,vGl:01}. An understanding of what makes processes
equivalent is important for several applications, ranging from comparison
and minimisation in model checking to program construction using
abstraction and refinement. Several equivalence relations have been
studied, most notably \emph{observation equivalence}~\cite{Mil:89},
\emph{failures equivalence}~\cite{Hoa:85}, and \emph{trace
equivalence}~\cite{Hoa:85}. Each equivalence has its own properties, making
it suitable for particular applications and verification
tasks~\cite{vGl:01}.

This paper focuses on \emph{conflict equivalence}, which compares processes
based on which other processes they can come into
conflict~\cite{CasLaf:99,RamWon:89} with. Two
processes are in conflict, if they can reach a state from which termination
is no longer possible. This can be because of \emph{deadlock} where
neither process is capable of doing anything, or \emph{livelock} where
the system continues to run without ever terminating.

It is difficult to reason about conflicts in a modular way. If two
processes are free from conflict individually, they may well be involved in
a conflict when running together, and vice versa~\cite{WonThiMalHoa:00}.
This makes it difficult to apply most methods of abstraction common in
model checking~\cite{BaiKat:08} to verify systems to be free from conflict,
and standard process-algebraic equivalences~\cite{vGl:01} are not
applicable either.

Conflict equivalence is introduced in~\cite{MalStrRee:06} as the
best possible process equivalence to reason compositionally about
conflicts. Conflict equivalence is coarser than observation
equivalence~\cite{Mil:89} and different from failures and trace
equivalence~\cite{Hoa:85}. The process-algebraic theory most closely
related to conflict equivalence is \emph{fair
testing}~\cite{BriRenVog:95,RenVog:07,NatCle:95}. The essential difference
between conflict equivalence and fair testing lies in the capability to
compare processes that exhibit blocking behaviour, as expressed by the
\emph{set of certain conflicts}~\cite{Mal:04,MalStrRee:06,Mal:10}.

In~\cite{FloMal:09,WarMal:10,SuSchRooHof:10}, various conflict-preserving
rewrite rules are used to simplify processes and check whether or not
large systems of concurrent finite-state automata are free from conflict.
While of good use in practice, the rewrite rules are incomplete, and it
remains an open question how processes can be normalised or compared for
conflict equivalence.

This paper improves on previous results about conflict equivalence and the
associated conflict preorder~\cite{MalStrRee:06}, and fair
testing~\cite{RenVog:07}, by providing a state-based characterisation of
the conflict preorder. It proposes \emph{less conflicting pairs} as a more
concrete way to compare processes for their conflicting behaviour than the
abstract test-based characterisation using \emph{nonconflicting
completions} in~\cite{MalStrRee:06} and the \emph{refusal trees}
of~\cite{RenVog:07}. Less conflicting pairs give a means to directly
compare processes based on their reachable state sets, which leads to an
alternative algorithm to test the conflict preorder. While still linear
exponential, this algorithm is simpler and has better time complexity than
the decision procedure for fair testing~\cite{RenVog:07}.

In the following, \sect~\ref{sec:preliminaries} briefly reviews the needed
terminology of languages, automata, and conflict equivalence. Then
\sect~\ref{sec:LC} introduces less conflicting pairs and shows how they can
be used to describe certain conflicts and the conflict preorder.
Afterwards, \sect~\ref{sec:MC} proposes an algorithm to calculate less
conflicting pairs for finite-state automata, and
\sect~\ref{sec:conclusions} adds some concluding remarks.


\section{Preliminaries}
\label{sec:preliminaries}

\subsection{Languages and Automata}

Event sequences and languages are a simple means to describe process
behaviours. Their basic building blocks are \emph{events}, which are taken
from a finite \emph{alphabet}~$\ACT$. Two special events are used, the
\emph{silent event}~$\tau$ and the \emph{termination event}~$\terminate$.
These are never included in an alphabet~$\ACT$ unless mentioned explicitly.

$\ACTstar$ denotes the set of all finite \emph{traces} of the form
$\ev_1 \ev_2 \cdots \ev_n$ of events from~$\ACT$, including the
\emph{empty trace}~$\varepsilon$.
The \emph{length} of trace~$s$ is denoted by $|s|$.
A subset $L \subseteq \ACTstar$ is called a \emph{language}.
The \emph{concatenation} of two traces $s,t \in \ACTstar$ is written
as~$st$, and a trace $s$ is called a \emph{prefix} of~$t$, written $s
\prefix t$, if $t=su$ for some trace~$u$. A language $L \subseteq \ACTstar$
is \emph{prefix-closed}, if $s \in L$ and $r \prefix s$ implies $r \in L$.

In this paper, process behaviour is modelled using nondeterministic
\emph{labelled transitions systems} or \emph{automata}
$A = \auttuple$, where
$\ACT$ is a finite alphabet of \emph{events},
$\Q$ is a set of \emph{states},
$\intrans \subseteq \Q \times (\ACT \cup \{\tau,\terminate\}) \times \Q$
is the \emph{state transition relation},
and $\Qi \subseteq \Q$ is the set of \emph{initial states}.
The automaton~$A$ is called \emph{finite-state} if its state set~$\Q$ is
finite.

The transition relation is written in infix notation $\state \trans[\ev]
\altstate$, and is extended to traces by
letting $\state \trans[\varepsilon] \state$ for all $\state \in \Q$, and
$\state \trans[s\ev] \altstate$ if $\state \trans[s] z \trans[\ev]
\altstate$ for some $z \in \Q$.
The transition relation must satisfy the additional
requirement that, whenever $\state \trans[\terminate] \altstate$, there
does not exist any outgoing transition from~$\altstate$.
The automaton~$A$ is called \emph{deterministic} if $|\initstateset| \leq 1$
and the transition relation contains no transitions labelled~$\tau$, and if
$x \trans[\ev] y_1$ and $x \trans[\ev] y_2$ always implies $y_1 =
y_2$.

To support silent transitions, $\state \ttrans[s] \altstate$, with $s \in
(\ACT \cup \{\terminate\})^*$, denotes the existence of a trace $t \in
(\ACT \cup \{\terminate,\tau\})^*$ such that $\state \trans[t] \altstate$,
and $s$ is obtained from~$t$ by deleting all $\tau$~events.
For a state set $X \subseteq \Q$ and a state
$y \in \Q$, the expression $X \ttrans[s] y$ denotes the existence of
$\state \in X$ such that $\state \ttrans[s] \altstate$, and $A \ttrans[s]
\altstate$ means that $\initstateset \ttrans[s] \altstate$. Furthermore,
$\state \ttrans \altstate$ denotes the existence of a trace~$s$ such that
$\state \ttrans[s] \altstate$, and $\state \ttrans[s]$ denotes the
existence of a state $\altstate \in \Q$ such that $\state \ttrans[s]
\altstate$. For a state, state set, or automaton~$\XX$, the \emph{language}
and the \emph{marked language} are
\begin{equation}
  \LANG(\XX) = \{\, s \in (\ACT\cup\{\terminate\})^* \mid \XX \ttrans[s] \,\}
  \qquad\mbox{and}\qquad
  \LANGM(\XX) = \LANG(\XX) \cap \ACTstar\terminate \ .
\end{equation}
Every prefix-closed language~$L$ is recognised by an automaton~$A$ such
that $\LANG(A) = L$, but only \emph{regular} languages are recognised by a
finite-state automaton~\cite{HopMotUll:01}.

When two automata are running in parallel, lock-step synchronisation in the
style of~\cite{Hoa:85} is used. The \emph{synchronous composition} of $A =
\langle\ACT_A\initphant\bcom \Q[A]\initphant\bcom \intrans\initphant_A\bcom
\Qi[A]\rangle$ and $B = \langle\ACT_B\initphant\bcom \Q[B]\initphant\bcom
\intrans\initphant_B\bcom \Qi[B]\rangle$ is
\begin{equation}
  A \sync B = \langle \ACT_A \cup \ACT_B, \Q[A] \times \Q[B],
  \intrans, \Qi[A] \times \Qi[B] \rangle
\end{equation}
where
$$
  \begin{array}{@{}r@{\quad}l@{}}
    (\state[A],\state[B])\trans[\sigma](\altstate[A],\altstate[B]) &
    \mbox{if}\ \sigma \in (\ACT_A \cap \ACT_B) \cup \{\terminate\},\
    \state[A] \trans[\sigma]_A \altstate[A],\ \mbox{and}\
    \state[B] \trans[\sigma]_B \altstate[B] \ ; \\
    (\state[A],\state[B]) \trans[\sigma] (\altstate[A],\state[B]) &
    \mbox{if}\ \sigma \in (\ACT_A \setminus \ACT_B) \cup \{\tau\}\
    \mbox{and}\ \state[A] \trans[\sigma]_A \altstate[A] \ ; \\
    (\state[A],\state[B]) \trans[\sigma] (\state[A],\altstate[B]) &
    \mbox{if}\ \sigma \in (\ACT_B \setminus \ACT_A) \cup \{\tau\}\
    \mbox{and}\ \state[B] \trans[\sigma]_B \altstate[B] \ .
  \end{array}
$$
In synchronous composition, shared events (including~$\terminate$) must be
executed by all automata together, while events used by only one of the
composed automata and silent ($\tau$) events are executed independently.


\subsection{Conflict Equivalence}
\label{subConflicts}

The key liveness property in supervisory control theory~\cite{RamWon:89} is
the \emph{nonblocking} property. Given an automaton~$A$, it is desirable
that every trace in~$\LANG(A)$ can be completed to a trace in~$\LANGM(A)$,
otherwise $A$ may become unable to terminate.
A process that may become unable to terminate is called
\emph{blocking}. This concept becomes more interesting when multiple
processes are running in parallel---in this case the term
\emph{conflicting} is used instead.

\begin{definition}
  An automaton $A = \auttuple$ is \emph{nonblocking} if for every state
  $\state \in Q$, $\Qi \ttrans \state$ implies that $\LANGM(\state) \neq
  \emptyset$. Otherwise $A$ is \emph{blocking}. Two automata $A$ and~$B$
  are \emph{nonconflicting} if $A \sync B$ is nonblocking, otherwise they
  are \emph{conflicting}.
\end{definition}

\begin{example}
\begin{figure}
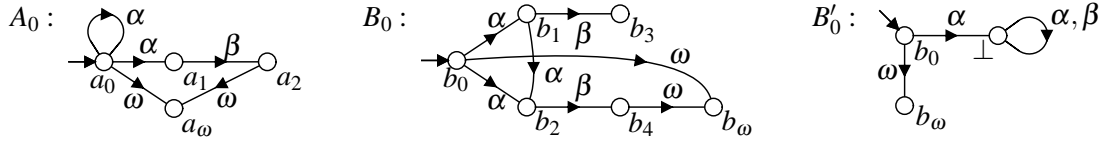

  \psfrag{a}{$\alpha$}
  \psfrag{ab}{$\alpha,\beta$}
  \psfrag{b}{$\beta$}
  \psfrag{o}{$\terminate$}
  \psfrag{a0}{$a_0$}
  \psfrag{a1}{$a_1$}
  \psfrag{a2}{$a_2$}
  \psfrag{aw}{$a_\terminate$}
  \psfrag{b0}{$b_0$}
  \psfrag{b1}{$b_1$}
  \psfrag{b2}{$b_2$}
  \psfrag{b3}{$b_3$}
  \psfrag{b4}{$b_4$}
  \psfrag{bw}{$b_\terminate$}
  \psfrag{dump}{$\bot$}
  \centering
  \autbox{A0}{$A_0:$}\qquad
  \autbox{B0}{$B_0:$}\qquad
  \autbox{B0conf}{$B'_0:$}%
  \caption{Examples of blocking and nonblocking automata.}
  \label{fig:blocking}
\end{figure}
  Automaton~$A_0$ in \fig~\ref{fig:blocking} is nonblocking, as it is
  always possible to reach state~$a_2$ and terminate.
  Automaton~$B_0$ on the other hand  is blocking, because it can
  enter state~$b_3$ after execution of $\alpha\beta$, from where it is no
  longer possible to reach a state where the termination event~\terminate\
  is enabled.
\end{example}

For an automaton to be nonblocking, it is enough that a terminal state
\emph{can} be reached from \emph{every} reachable state. There is no
requirement for termination to be guaranteed. For example, automaton~$A_0$
in \fig~\ref{fig:blocking} is nonblocking despite the presence of a
possibly infinite loop of $\alpha$-transitions in state~$a_0$. Nonblocking
is also different from ``may''-testing~\cite{RenVog:07}, which only
requires the possibility of termination from the initial state. The testing
semantics most similar to nonblocking is ``should''-testing, which is also
known as \emph{fair testing}~\cite{RenVog:07}.

To reason about nonblocking in a compositional way, the notion of
\emph{conflict equivalence} is developed in~\cite{MalStrRee:06}. According
to process-algebraic testing theory, two automata are considered as
equivalent if they both respond in the same way to all tests of a certain
type~\cite{DeNHen:84}. For conflict equivalence, a \emph{test} is an
arbitrary automaton, and the \emph{response} is the observation whether or
not the test is conflicting with the automaton in question.

\begin{definition}
  \label{def:confeq}
  Let $A$ and~$B$ be two automata.
  $A$~is \emph{less conflicting} than~$B$, written $A \confle B$, if, for
  every automaton~$T$, if $B \sync T$ is nonblocking then $A \sync T$ also
  is nonblocking. $A$ and~$B$ are \emph{conflict equivalent}, $A \confeq
  B$, if $A \confle B$ and $B \confle A$.
\end{definition}

\begin{figure}
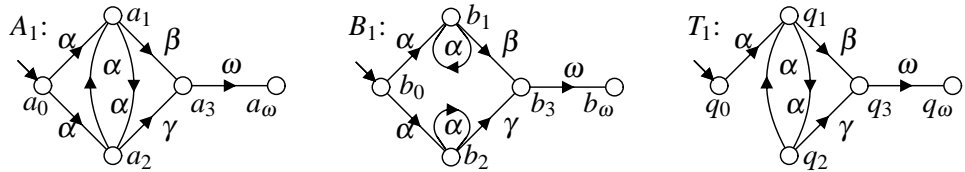

  \centerline{%
  \psfrag{a}{$\alpha$}
  \psfrag{b}{$\beta$}
  \psfrag{c}{$\gamma$}
  \psfrag{o}{\terminate}
  \tabcolsep1em
  \begin{tabular}{lll}
    $A_1$: & $B_1$: & $T_1$: \\
    \noalign{\vskip-4ex}
    \psfrag{a0}{$a_0$}%
    \psfrag{a1}{$a_1$}%
    \psfrag{a2}{$a_2$}%
    \psfrag{a3}{$a_3$}%
    \psfrag{aw}{$a_\terminate$}%
    \autgraphics{A1} &
    \psfrag{b0}{$b_0$}%
    \psfrag{b1}{$b_1$}%
    \psfrag{b2}{$b_2$}%
    \psfrag{b3}{$b_3$}%
    \psfrag{bw}{$b_\terminate$}
    \autgraphics{B1} &
    \psfrag{S0}{$q_0$}%
    \psfrag{S1}{$q_1$}%
    \psfrag{S2}{$q_2$}%
    \psfrag{S3}{$q_3$}%
    \psfrag{qw}{$q_\terminate$}%
    \autgraphics{T1} \\
  \end{tabular}}
  \caption{Two automata that are not conflict equivalent.}
  \label{fig:notconfeq}
\end{figure}

\begin{example}
  Consider automata $A_1$ and~$B_1$ in \fig~\ref{fig:notconfeq}. $A_1$~is
  \emph{not} less conflicting than~$B_1$, since $A_1 \sync T_1$ is blocking
  while $B_1 \sync T_1$ is nonblocking. This is because $A_1 \sync T_1$ can
  enter the blocking state~$(a_2,q_1)$ after executing of~$\alpha$, whereas
  after executing~$\alpha$ in~$B_1$, it eventually becomes possible to
  continue using either the $\beta$- or~$\gamma$-transition of~$T_1$. It
  can also be shown that $B_1 \confle A_1$ does not hold.
\end{example}

The properties of the conflict preorder~$\confle$ and of conflict
equivalence and their relationship to other process-algebraic relations are
studied in~\cite{MalStrRee:06}. It is enough to consider deterministic
tests in \defn~\ref{def:confeq}, and conflict equivalence is is the
coarsest possible congruence with respect to synchronous composition that
respects blocking, making it an ideal equivalence for use in compositional
verification~\cite{FloMal:09,WarMal:10}.

\subsection{The Set of Certain Conflicts}
\label{sub:CONF}

Every automaton can be associated with a language of \emph{certain
conflicts}, which plays an important role in conflict
semantics~\cite{Mal:04}.

\begin{definition}
  For an automaton $A = \auttuple$, write
  \begin{align}
    \CONF(A)  &= \{\, s \in \ACTstar  \mid
                      \mbox{For every automaton $T$ such that $T \ttrans[s]$,
                            $A \sync T$ is blocking} \,\}\ ; \\
    \NCONF(A) &= \{\, s \in \ACTstar \mid
                      \mbox{There exists an automaton $T$ such that
                            $T \ttrans[s]$ and $A \sync T$
                            is nonblocking} \,\}\ .
  \end{align}
\end{definition}

$\CONF(A)$ is the set of \emph{certain conflicts} of~$A$.
It contains all traces that, when possible in the environment,
necessarily cause blocking.
Its complement $\NCONF(A)$ is the most general behaviour of processes
that are to be nonconflicting with~$A$.
If $A$ is nonblocking, then $\CONF(A) = \emptyset$ and $\NCONF(A) =
\ACTstar$, because in this case $A \sync U$ is nonblocking, where $U$
is a deterministic automaton such that $\LANGM(U) = \ACTstar\terminate$.
The set of certain conflicts becomes more interesting for
blocking automata.

\begin{example}
  \label{ex:certainconf}
  Consider again automaton~$B_0$ in \fig~\ref{fig:blocking}.
  Clearly $\alpha\beta \in \CONF(B_0)$ as $B_0$ can enter the deadlock
  state~$b_3$ by executing~$\alpha\beta$, and therefore every test~$T$
  that can execute~$\alpha\beta$ is conflicting with~$B_0$.
  But also $\alpha \in \CONF(B_0)$, because $B_0$ can enter state~$b_2$
  by executing~$\alpha$, from where the only possibility to terminate
  is by executing~$\beta\terminate$. So any test that can execute~$\alpha$
  also needs to be able to execute~$\alpha\beta$ if it is to be
  nonconflicting with~$B_0$; but such a test is conflicting with~$B_0$ as
  explained above. It can be shown that $\CONF(B_0) = \alpha\ACTstar$.
\end{example}

The set of certain conflicts is introduced in~\cite{Mal:04}, and its
properties and its relationship to conflict equivalence are studied
in~\cite{MalStrRee:06}. Even if an automaton is nondeterministic, its set
of certain conflicts is a \emph{language}, but as shown in
\examp~\ref{ex:certainconf}, it is not necessarily a subset of the
language~$\LANG(A)$ of its automaton. If a trace~$s$ is a trace of certain
conflicts, then so is any extension~$st$. An algorithm to compute the set
of certain conflicts for a given finite-state automaton is presented
in~\cite{Mal:10}.

Certain conflicts constitute the main difference between conflict
equivalence and \emph{fair testing}~\cite{RenVog:07}. In fair testing,
processes are not allowed to synchronise on the termination
event~\terminate, so termination is determined solely by the test. This can
be expressed as conflict equivalence by requiring that \terminate\ be
enabled in all states of the automata compared~\cite{MalStrRee:06}.

Conversely, it is possible to factor out certain conflicts from any given
automaton, by redirecting all traces of certain conflicts to a single
state~\cite{Mal:04,Mal:10}. For example, automaton~$B_0$ in
\fig~\ref{fig:blocking} can be replaced by the conflict equivalent
automaton~$B'_0$, which uses the single deadlock state~$\bot$. Two automata
$A$ and~$B$ are conflict equivalent if and only if their normalised forms
$A'$ and~$B'$ are fair testing equivalent. The decision procedure for fair
testing~\cite{RenVog:07} can be used to test the conflict preorder, and
vice versa.


\section{Characterising the Conflict Preorder}
\label{sec:LC}

This section is concerned about characterising two automata $A$ and~$B$ as
conflict equivalent, or characterising $A$ as less conflicting than~$B$, in
a state-based way.
First, \ref{sub:understanding} explains the crucial properties of conflict
equivalence using examples. \emph{Less conflicting pairs} are introduced
in~\ref{sub:LC}, and they are used to characterise certain conflicts
in~\ref{sub:LC:CONF} and the conflict preorder in~\ref{sub:LC:confle}.

\subsection{Understanding Conflict Equivalence}
\label{sub:understanding}

Every reachable state of an automaton~$A$ carries a \emph{nonblocking
requirement} (also known as a \emph{nonconflicting
completion}~\cite{MalStrRee:06}) that needs to be satisfied by tests that
are to be nonconflicting with~$A$. For example, if $A \ttrans[s] x_A$, then
every test~$T$ that can execute~$s$ needs to be able to continue with at
least one trace $t \in \LANGM(x_A)$, or $T$ is conflicting with~$A$. An
automaton~$A$ is less conflicting than another automaton~$B$, if every
nonblocking requirement associated with~$A$ also is a nonblocking
requirement associated with~$B$.

\begin{figure}
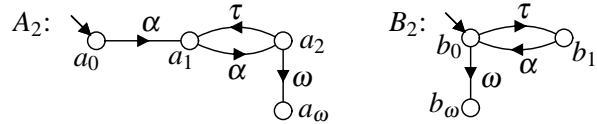

  \centerline{%
  \psfrag{a}{$\alpha$}
  \psfrag{t}{$\tau$}
  \psfrag{o}{\terminate}
  \tabcolsep1em
  \begin{tabular}{ll}
    $A_2$: & $B_2$: \\
    \noalign{\vskip-4ex}
    \psfrag{a0}{$a_0$}%
    \psfrag{a1}{$a_1$}%
    \psfrag{a2}{$a_2$}%
    \psfrag{aw}{$a_\terminate$}%
    \kern1.7em\autgraphics{A2} &
    \psfrag{b0}{$b_0$}%
    \psfrag{b1}{$b_1$}%
    \psfrag{bw}{$b_\terminate$}
    \quad\autgraphics{B2} \\
  \end{tabular}}
  \caption{Two automata that are conflict equivalent.}
  \label{fig:confeq}
\end{figure}

\begin{example}
  Consider again automata $A_1$ and~$B_1$ in \fig~\ref{fig:notconfeq}.
  They have the same marked languages.
  Thus, if the initial state~$a_0$ of~$A_1$ is
  blocking in combination with some test~$T$, then so is the initial
  state~$b_0$ of~$B_1$. But this is not the case when $A_1 \sync T$ enters
  a state $(a_1, x_T)$ after execution of~$\alpha$. State~$a_1$ requires
  $x_T$ to be capable of performing at least one trace from the language
  $\LANGM(a_1) = (\alpha\alpha)^*\beta\terminate +
  (\alpha\alpha)^*\alpha\gamma\terminate$, whereas the states $b_1$
  and~$b_2$, which can both be entered after executing~$\alpha$, require
  a trace from the language $\alpha^*\beta\omega$ and
  $\alpha^*\gamma\omega$, respectively. Both of these languages contain
  traces outside of the language~$\LANGM(a_1)$. Automaton~$T_1$ in
  \fig~\ref{fig:notconfeq} is in conflict with $A_1$ but not with~$B_1$.
\end{example}

In general, it is not enough to compare only the marked languages of states
reached by equal traces. Not every nonblocking requirements is a marked
language of some state of its automaton. The following example shows one of
the problems.

\begin{example}
  Consider automata $A_2$ and~$B_2$ in \fig~\ref{fig:confeq}. The marked
  language of the initial state of~$A_2$ is $\LANGM(a_0) =
  \alpha\alpha^+\terminate$, while the marked languages of the two states
  in~$B_2$ that can be entered initially are $\LANGM(b_0) =
  \alpha^*\terminate$ and $\LANGM(b_1) = \alpha^+\terminate$.
  Although the marked languages are different, for any
  automaton~$T$, if $B_2 \sync T$ is nonblocking, then $A_2 \sync T$ must
  also be nonblocking.
  If $T$ is to be nonconflicting in combination with~$B_2$, since $B_2$ may
  initially enter state~$b_1$, there must be the possibility to continue
  with event~$\alpha$. However, after executing $\alpha$, automaton $B_2$ may
  again silently enter state~$b_1$, which means that $\alpha$ must be possible
  again. This is enough to ensure that $A_2 \sync T$ is nonblocking.
  Using this argument, it can be shown that $A_2$ and~$B_2$ are conflict
  equivalent.
\end{example}

\subsection{Less Conflicting Pairs}
\label{sub:LC}

In order to compare two nondeterministic automata according to conflicts,
it is necessary to identify sets of states the two automata may reach
under the same input. This is done using the well-known \emph{subset
  construction}~\cite{HopMotUll:01}. To capture termination, the usual
powerset state space is extended by a special state~$\terminate$ entered
only after termination.

\begin{definition}
  \label{def:ddet}
  The \emph{deterministic state space} of automaton~$A = \auttuple$ is
  \begin{equation}
    \Qdet[A] = \power\Q \cup \{\terminate\} \ ,
  \end{equation}
  and the \emph{deterministic transition function} $\ddet[A]\colon \Qdet
  \times (\ACT \cup \{\terminate\}) \to \Qdet$ for~$A$ is defined as
  \begin{equation}
    \ddet[A](X,\ev) =
                      \begin{cases}
                         \terminate, & \mbox{if}\ \ev=\terminate\
                                      \mbox{and}\ X \ttrans[\terminate]; \\
                         \{\, y \in \stateset \mid X \ttrans[\sigma] y \,\},
                                    & \mbox{otherwise}.
                      \end{cases}
  \end{equation}
\end{definition}

The deterministic transition function~$\ddet[A]$ is extended to traces $s
\in \ACTstar \cup \ACTstar\terminate$ in the standard way. Note that
$\ddet[A](X,s)$ is defined for every trace $s \in \ACTstar \cup
\ACTstar\terminate$; if none of the states in~$X$ accepts the trace~$s$,
this is indicated by $\ddet[A](X,s) = \emptyset$. This is also true for
termination: if $\terminate$ is enabled in some state in~$X$, then
$\ddet[A](X,\terminate) = \terminate$, otherwise $\ddet[A](X,\terminate) =
\emptyset$.

In order to compare two automata $A$ and~$B$ with respect to possible
conflicts, \emph{pairs} of state sets of the subset construction of $A$
and~$B$ need to be considered. Therefore, the deterministic transition
function is also applied to pairs $\XX = (X_A,X_B)$ of state sets $X_A
\subseteq \Q[A]$ and $X_B \subseteq \Q[B]$,
\begin{equation}
  \ddet[A,B](\XX,s) = \ddet[A,B](X_A,X_B,s) =
  (\ddet[A](X_A,s), \ddet[B](X_B,s)) \ .
\end{equation}

To determine whether $A \confle B$, it is necessary to check all states
$x_A \in \Q[A]$ against matching state sets $X_B \subseteq \Q[B]$ and
determine whether all possible conflicts of~$x_A$ are also present
in~$X_B$. For example, when automaton~$A_2$ in \fig~\ref{fig:confeq} is in
state~$a_1$, then $B_2$ may be in $b_0$ or~$b_1$. In
state~$a_1$, at least one of the traces in $\alpha^+\terminate$ needs to be
enabled to avert blocking, and the same requirement to avert blocking is
seen in state~$b_1$. When state~$a_1$ is entered with some test~$T$,
blocking occurs if none of the traces in~$\alpha^+\terminate$
is enabled, and such a test~$T$ is also blocking when combined with a system
that may be in $b_0$ or~$b_1$. Therefore, $a_1$ is considered
in the following as \emph{less conflicting}~(\LC) than~$\{b_0,b_1\}$.

It cannot always be determined directly whether a state $x_A \in \Q[A]$ is
less conflicting than a state set $X_B \subseteq \Q[B]$. In some cases, it
is necessary also to consider the deterministic successors of $x_A$ and~$X_B$.
Therefore, the following definition considers pairs~$(X_A,X_B)$ of state
sets.

\begin{definition}
  \label{def:LC}
  Let $A = \auttuple[A]$ and $B = \auttuple[B]$ be automata. The set
  $\LC(A,B) \subseteq \Qdet[A] \times \Qdet[B]$ of \emph{less conflicting
  pairs} for $A$ and~$B$ is inductively defined by
  \begin{align}
    \label{eq:LC:0}
    \LC^0(A,B) &= \{\terminate\} \times \Qdet[B] \;\cup\;
                  \{\, (X_A,X_B) \mid X_B \subseteq \stateset[B]\
                       \mbox{and there exists}\
                       x_B \in X_B\ \mbox{with}\
                       \LANGM(x_B) = \emptyset\,\}\ ; \\
    \label{eq:LC:succ}
    \LC^{n+1}(A,B) &=
       \LongSet{11em}{$(X_A,X_B)$ $\mid$
                     there exists $x_B \in X_B$ such that for all
                     $t \in \ACTstar$, if $x_B \ttrans[t\terminate]$
                     then there exists $r \prefix t\terminate$
                     such that $\ddet[A,B](X_A, X_B, r) \in \LC^i(A,B)$
                     for some $i \leq n$}; \\
    \label{eq:LC:union}
    \LC(A,B) &= \bigcup_{n \geq 0} \LC^n(A,B)\ .
  \end{align}
\end{definition}

\begin{remark}
  \label{rem:notLC}
  If $(X_A,X_B) \notin \LC(A,B)$,
  then according to~\eqref{eq:LC:succ},
  for every state $x_B \in X_B$, there exists $t \in \ACTstar$
  such that $x_B \ttrans[t\terminate]$, and $\ddet(X_A,X_B,r) \notin
  \LC(A,B)$ for all prefixes $r \prefix t\terminate$.
\end{remark}

The idea of \defn~\ref{def:LC} is to classify a pair~$(X_A,X_B)$ as less
conflicting, if the marked language of~$X_A$ is a \emph{nonconflicting
completion}~\cite{MalStrRee:06} for the process with initial states~$X_B$.
That is, every test that is nonconflicting in combination with each of the
states in~$X_B$ can terminate with at least one trace from the marked
language of~$X_A$. Or conversely, every test that cannot terminate using
any of the traces in the marked language of~$X_A$ also is conflicting
with~$X_B$ (see \lemm~\ref{lem:LC1} below).

The first state set~$X_A$ of a pair $(X_A,X_B)$ is just used to represent a
\emph{language} of possible completions. If state sets $X_A$ and~$Y_A$ have
the same languages, then all pairs $(X_A,X_B)$ and~$(Y_A,X_B)$ have exactly
the same less conflicting status. For the second state set~$X_B$ on the
other hand, the complete nondeterministic behaviour is relevant.

A pair~$(\terminate,X_B)$ is considered as ``less
conflicting''~\eqref{eq:LC:0}, since termination has already been achieved
in~$A$. If $X_B$ contains a state $x_B$ such that
$\LANGM(x_B) = \emptyset$, then $(X_A,X_B)$ also is less
conflicting~\eqref{eq:LC:0}, because conflict is guaranteed in~$X_B$. For
other pairs~$(X_A,X_B)$, it must be checked whether $X_B$ contains a
requirement to avert blocking matching that given by the language
of~$X_A$~\eqref{eq:LC:succ}.

\begin{example}
\label{ex:LC:0}
Consider again automata $A_0$ and~$B_0$ in \fig~\ref{fig:blocking}. It
holds that $(\{a_0\}, \{b_0\}) \in \LC^1(A_0,B_0)$. There are three ways
to terminate from~$b_0$, by executing \terminate\ or
$\alpha\beta\terminate$ or~$\alpha\alpha\beta\terminate$. All three traces
are possible in~$a_0$, each taking the pair $(\{a_0\}, \{b_0\})$ to the
deterministic successor $(\terminate,\terminate) \in \LC^0(A_0,B_0)$. This
is enough to confirm that \eqref{eq:LC:succ} is satisfied.

On the other hand, $(\{a_0\}, \{b_2\}) \notin \LC^1(A_0,B_0)$. From
state~$a_0$, blocking occurs with a test~$T$ that can only execute
$\beta\terminate$, but this test is nonblocking with~$b_2$. It holds
that $b_2 \trans[\beta\terminate]$, where trace $\beta\terminate$ has the
prefixes $\varepsilon$, $\beta$, and~$\beta\terminate$, but
$\ddet[A_0,B_0](\{a_0\}, \{b_2\}, \varepsilon) = (\{a_0\}, \{b_2\}) \notin
\LC^0(A_0,B_0)$, $\ddet[A_0,B_0](\{a_0\}, \{b_2\}, \beta) =
(\emptyset,\{b_4\}) \notin \LC^0(A_0,B_0)$, and $\ddet[A_0,B_0](\{a_0\},
\{b_2\}, \beta\terminate) = (\emptyset,\terminate) \notin \LC^0(A_0,B_0)$.
Therefore, \eqref{eq:LC:succ}~is not satisfied and $(\{a_0\}, \{b_2\})
\notin \LC^1(A_0,B_0)$. It can also be shown that $(\{a_0\}, \{b_2\})
\notin \LC(A_0,B_0)$.
\end{example}

For a \emph{level-1} less conflicting pair $(X_A,X_B) \in \LC^1(A,B)$, if
$X_B$ does not contain blocking states, then there must exist a state $x_B
\in X_B$ such that $\LANGM(x_B) \subseteq \LANGM(X_A)$. This is not the
case for every less conflicting pair, as some nonblocking requirements are
only implicitly contained in the automaton. To show that $(X_A,X_B)$ is a
less conflicting pair, it is enough to find a state in~$x_B \in X_B$ that
can cover an initial segment of~$\LANGM(X_A)$, as long as a less
conflicting pair of a \emph{lower level} is reached afterwards.

\begin{example}
  \label{ex:LC:succ}
  Consider again automata $A_2$ and~$B_2$ in \fig~\ref{fig:confeq}. By
  definition, $(\terminate,\terminate) \in \LC^0(A_2,B_2)$, and following
  from this, $(\{a_1\}, \{b_0,b_1\}) \in \LC^1(A_2,B_2)$, because the
  marked language of~$a_1$ is $\alpha^+\terminate$, which also is the
  marked language of~$b_1$.

  Now consider the pair $(\{a_0\}, \{b_0,b_1\})$. State~$a_0$ has the marked
  language $\alpha\alpha^+\terminate$, i.e., to avert
  blocking from~$a_0$, a test must be able to execute at least one of the
  traces in~$\alpha\alpha^+\terminate$. Although this language is not
  directly associated with any state in~$B_2$, the nonblocking requirement
  is implicitly present in state~$b_1$. If blocking is to be averted from
  state~$b_1$, event~$\alpha$ must be possible. After executing~$\alpha$,
  state~$b_0$ is entered, from where it is always possible to silently
  return to state~$b_1$ with marked language~$\alpha^+\terminate$.
  Therefore, in order to avert blocking from state~$b_1$, it is necessary
  to execute~$\alpha$ and afterwards be able to terminate using one of the
  traces in~$\alpha^+\terminate$. This amounts to the implicit nonblocking
  requirement to execute a trace from~$\alpha\alpha^+\terminate$ in
  state~$b_1$.

  Therefore $(\{a_0\}, \{b_0,b_1\}) \notin \LC^1(A_2,B_2)$, but $(\{a_0\},
  \{b_0,b_1\}) \in \LC^2(A_2,B_2)$ according to~\eqref{eq:LC:succ}: every
  trace that leads to a terminal state from state~$b_1$ has the
  prefix~$\alpha$, and $\ddet[A_2,B_2](\{a_0\}, \{b_0,b_1\}, \alpha) =
  (\{a_1\}, \{b_0,b_1\}) \in \LC^1(A_2,B_2)$.
\end{example}

As shown in the example, some nonblocking requirements have to be
constructed using a saturation operation that combines two previously found
nonblocking requirements. The level~$n$ of a less conflicting pair
$(X_A,X_B) \in \LC^n(A,B)$ represents the nesting depth of applications of
this saturation operation.


The following two lemmas relate the state-based definition of less
conflicting pairs to possible tests and thus to the conflict preorder.
A pair~$(X_A,X_B)$ is a less conflicting pair, if every test~$T$ such that
$\LANGM(X_A) \cap \LANGM(T) = \emptyset$ also is conflicting with~$X_B$.

\begin{lemma}
  \label{lem:LC1}
  Let $A = \auttuple[A]$, $B = \auttuple[B]$, and
  $T = \auttuple[T]$ be automata, and let $x_T \in \stateset[T]$
  be a (possibly unreachable) state.
  For every less conflicting pair $(X_A,X_B) \in \LC(A,B)$,
  at least one of the following conditions holds.
  \begin{enumerate}
  \item \label{it:LC1:nonblocking1}
        $X_A = \terminate$, or
        $X_A \subseteq \stateset[A]$ and there exists $x_A \in X_A$
        such that $\LANGM(x_A,x_T) \neq \emptyset$.
  \item \label{it:LC1:blocking2}
        There exist states $x_B \in X_B$, $y_B \in \Q[B]$, and
	$y_T \in \Q[T]$ such that
        $(x_B,x_T) \ttrans (y_B,y_T)$ and $\LANGM(y_B,y_T) = \emptyset$.
  \end{enumerate}
  (Here and in the following, notation $\LANGM(x_A, x_T)$ is abused to be a
  shorthand for $\LANGM((x_A, x_T))$.)
\end{lemma}

\begin{proof}
  As $(X_A,X_B)$ is a less conflicting pair,
  it holds that $(X_A,X_B) \in \LC^n(A,B)$ for some $n \in \NAT$.
  The claim is shown by induction on~$n$.

  If $(X_A,X_B) \in \LC^0(A,B)$ then by~\eqref{eq:LC:0} it holds that
  $X_A = \terminate$, or $X_B \subseteq \stateset[B]$ and there exists
  $x_B \in X_B$ such that $\LANGM(x_B) = \emptyset$. In the first case
  \refi{it:LC1:nonblocking1} holds, and in the second case
  \refi{it:LC1:blocking2} holds as $(x_B,x_T) \trans[\varepsilon] (x_B,x_T)$
  and $\LANGM(x_B,x_T) = \LANGM(x_B) \cap \LANGM(x_T) = \emptyset$.

  Now assume the claim holds for all $i \leq n$, i.e., for all
  $(X_A,X_B) \in \LC^i(A,B)$, one of the conditions
  \refi{it:LC1:nonblocking1} or~\refi{it:LC1:blocking2} holds,
  and consider $(X_A,X_B) \in \LC^{n+1}(A,B)$.
  By~\eqref{eq:LC:succ}, there exists $x_B \in X_B$ such that for all $t
  \in \ACTstar$, if $x_B \ttrans[t\terminate]$ then there exists a prefix
  $r \prefix t\terminate$ such that $\ddet[A,B](X_A, X_B, r) \in \LC^i(A,B)$
  for some $i \leq n$.
  If $\LANGM(x_B,x_T) = \emptyset$, \refi{it:LC1:blocking2}~follows immediately
  as $(x_B,x_T) \trans[\varepsilon] (x_B,x_T)$.
  Therefore assume that $\LANGM(x_B,x_T) \neq \emptyset$,
  i.e., there exists $t \in \ACTstar$ such that $(x_B,x_T)
  \ttrans[t\terminate]$. Then $x_B \ttrans[t\terminate]$,
  so there exists $r \prefix t\terminate$ such that $\ddet[A,B](X_A, X_B, r)
  \in  \LC^i(A,B)$ for some $i \leq n$.
  As $r \prefix t\terminate$ and $x_T \ttrans[t\terminate]$,
  it also holds that $x_T \ttrans[r] y_T$ for some $y_T \in \stateset[T]$.
  Let $\ddet[A,B](X_A,X_B,r) = (Y_A,Y_B)$.
  By inductive assumption, \refi{it:LC1:nonblocking1}
  or~\refi{it:LC1:blocking2} holds for $(Y_A, Y_B) \in \LC^i(A,B)$
  and~$y_T$.

  \refi{it:LC1:nonblocking1}
  In this case, either $Y_A = \terminate$, or $Y_A \subseteq \stateset[A]$
  and there exists $y_A \in Y_A$ and $u \in \ACTstar$ such that
  $(y_A,y_T) \ttrans[u\terminate]$.
  If $Y_A = \terminate$, then
  $\ddet[A](X_A,r) = Y_A = \terminate$ and according to \defn~\ref{def:ddet}
  there exists $r_A \in \ACTstar$ such that $r = r_A\terminate$,
  and there exist states $x_A \in X_A$ and $y_A \in \Q[A]$ such that
  $x_A \ttrans[r_A] y_A \ttrans[\terminate]$, i.e.,
  $(x_A,x_T) \ttrans[r_A\terminate]$.
  If there exists $y_A \in Y_A$ and $u \in
  \ACTstar$ such that $(y_A,y_T) \ttrans[u\terminate]$, then
  since $\ddet[A](X_A,r) = Y_A$, there exists $x_A \in X_A$ such that
  $x_A \ttrans[r] y_A$, i.e.,
  $(x_A,x_T) \ttrans[r] (y_A,y_T) \ttrans[u\terminate]$.
  In both cases, \refi{it:LC1:nonblocking1}~holds for $(X_A,X_B)$ and~$x_T$.

  \refi{it:LC1:blocking2} If there exists a state $y_B \in Y_B$ such that
  $(y_B, y_T) \ttrans (z_B, z_T)$ where $\LANGM(z_B,z_T) = \emptyset$, then since
  $\ddet[B](X_B,r) = Y_B$, there exists $x_B \in X_B$ such that $x_B
  \ttrans[r] y_B$, which implies $(x_B,x_T) \ttrans[r] (y_B,y_T) \ttrans
  (z_B, z_T)$ with $\LANGM(z_B,z_T) = \emptyset$. Thus,
  \refi{it:LC1:blocking2}~holds for $(X_A,X_B)$ and~$x_T$.
\end{proof}

Conversely, if a pair of state sets is \emph{not} a less conflicting pair
for $A$ and~$B$, then this pair gives rise to a test automaton to show that
$A$ is not less conflicting than~$B$. This test exhibits blocking behaviour
in combination with~$A$ but not with~$B$.

\begin{lemma}
  \label{lem:LC2}
  Let $A = \auttuple[A]$ and $B = \auttuple[B]$ be automata.
  For every pair $\XX = (X_A,X_B) \notin \LC(A,B)$,
  there exists a deterministic automaton $T_\XX = \detauttuple[T]$
  such that both the following conditions hold.
  \begin{enumerate}
  \item \label{it:LC2:G1}
    For all states $x_A \in X_A$, it holds that $\LANGM(x_A\initphant,
    \initstate[T]) = \emptyset$.
  \item \label{it:LC2:G2}
    For all states $x_B \in X_B$, $y_B \in \Q[B]$, $y_T \in \stateset[T]$
    such that $(x_B\initphant, \initstate[T]) \ttrans (y_B, y_T)$, it holds
    that $\LANGM(y_B, y_T) \neq \emptyset$.
  \end{enumerate}
\end{lemma}

\begin{proof}
  Construct the deterministic automaton $T_\XX = \detauttuple[T]$ such that
  \begin{equation}
    \label{eq:LC2:T}
    \LANG(T_\XX) = \{\, s \in \ACTstar \cup \ACTstar\terminate \mid
                        \ddet[A,B](\XX,r) \notin \LC(A,B)\
                        \mbox{for all}\ r \prefix s \,\}\ .
  \end{equation}
  This language is prefix-closed by construction and nonempty because $\XX
  \notin \LC(A,B)$. Therefore, $T_\XX$ is a well-defined automaton.

  \refi{it:LC2:G1}
  Let $x_A \in X_A$. If $x_A \ttrans[t\terminate]$ for some $t \in
  \ACTstar$, then $\ddet[A,B](\XX,t\terminate) = (\terminate,Y_B) \in
  \LC^0(A,B) \subseteq \LC(A,B)$ for some $Y\initphant_B \in \Qdet[B]$ by
  \defn\ \ref{def:ddet} and~\ref{def:LC}.
  It follows from~\eqref{eq:LC2:T} that $t\terminate
  \notin \LANG(T_\XX)$, and thus $(x_A,\initstate[T]) \ttrans[t\terminate]$
  does not hold. Since $t \in \ACTstar$ was chosen arbitrarily, it follows
  that $\LANGM(x_A\initphant,\initstate[T]) = \emptyset$.

  \refi{it:LC2:G2}
  Let $x_B \in X_B$, $y_B \in \Q[B]$, $y_T \in \stateset[T]$, and $s \in
  \ACTstar$ such that $(x_B\initphant, \initstate[T]) \ttrans[s] (y_B, y_T)$.
  Clearly $s \in \LANG(T_\XX)$, and by~\eqref{eq:LC2:T} it follows that
  $\ddet[A,B](\XX,r) \notin \LC(A,B)$ for all prefixes $r \prefix s$.
  Let $\ddet[A,B](\XX,s) = \YY$.
  Then $\YY \notin \LC(A,B)$,
  so there exists a trace $t \in \ACTstar$ such that $y_B \ttrans[t\omega]$
  and for all $r \prefix t$ it holds that $\ddet[A,B](\YY, r)
  \notin \LC(A,B)$ (see \rem~\ref{rem:notLC}).
  Thus $x_B \ttrans[s] y_B \ttrans[t\terminate]$ and
  for all prefixes $u \prefix st\terminate$, it holds that
  $\ddet[A,B](\XX,u) \notin \LC(A,B)$.
  Then $st\terminate \in \LANG(T_\XX)$ according to~\eqref{eq:LC2:T},
  and since $T_\XX$ is deterministic, it follows that $y_T
  \ttrans[t\terminate]$.
  Therefore, $(y_B,y_T) \ttrans[t\terminate]$,
  i.e., $\LANGM(y_B, y_T) \neq \emptyset$.
\end{proof}

\subsection{Less Conflicting Pairs and Certain Conflicts}
\label{sub:LC:CONF}

Less conflicting pairs can be used to characterise the set of \emph{certain
conflicts} of an automaton as defined in~\ref{sub:CONF}. This shows the
close link between the conflict preorder and the set of certain conflicts.
If a pair~$(\emptyset,X_B)$ is a less conflicting pair then, since
termination is impossible from~$\emptyset$, conflict must be also present
in~$X_B$. In this case, every trace leading to~$X_B$ must be a trace of
certain conflicts. This observation leads to the following alternative
characterisation of the set of certain conflicts.

\begin{theorem}
  \label{thm:LC:CONF}
  The set of certain conflicts of $B = \auttuple$ can also be written as
  \begin{equation}
    \CONF(B) = \{\, s \in \ACTstar \mid
                    (\emptyset, \ddet[B](\initstateset,r)) \in \LC(O,B)\
                    \mbox{for some prefix}\ r \prefix s \,\}\ ,
  \end{equation}
  where $O = \langle\ACT, \emptyset, \emptyset, \emptyset\rangle$
  stands for the empty automaton.
\end{theorem}

\begin{proof}
  First let $s \in \ACTstar$ such that $(\emptyset,
  \ddet[B](\initstateset,r)) \in \LC(O,B)$ for some $r \prefix s$, and let
  $T = \auttuple[T]$ be an automaton such that $T \ttrans[s]$. It is to be
  shown that $B \sync T$ is blocking. Since $T \ttrans[s]$ and $r \prefix
  s$, it holds that $T \ttrans[r] x_T$ for some state $x_T \in \Q[T]$.
  Since $(\emptyset, \ddet[B](\initstateset,r)) \in \LC(O,B)$, either
  \refi{it:LC1:nonblocking1} or \refi{it:LC1:blocking2} in
  \lemm~\ref{lem:LC1} holds. However, \refi{it:LC1:nonblocking1} is
  impossible as the first state set of the pair is empty, so
  \refi{it:LC1:blocking2} must be true. Thus, there exists a state $x \in
  \ddet[B](\initstateset, r)$ such that $(x, x_T) \ttrans (y, y_T)$ where
  $\LANGM(y, y_T) = \emptyset$. Then $B \sync T$ is blocking as $B \sync T
  \ttrans[r] (x, x_T) \ttrans (y, y_T)$.

  Conversely, let $s \in \ACTstar$ such that $(\emptyset,
  \ddet[B](\initstateset,r)) \notin \LC(O,B)$ for every prefix $r \prefix
  s$. It is to be shown that $s \in \NCONF(B)$. Consider the deterministic
  automaton $T$ such that
  \begin{equation}
    \label{eq:LC:CONF}
    \LANG(T) = \{\, t \in \ACTstar \mid
                   (\emptyset,\ddet[B](\initstateset, r)) \notin \LC(O,B)\
                   \mbox{for all}\ r \prefix t \,\}\ .
  \end{equation}
  $T$ is a well-defined automaton as $\LANG(T)$ is prefix-closed by
  construction. It remains to be shown that
  $B \sync T$ is nonblocking. Let $B \sync T \ttrans[t] (x, x_T)$.
  Then $t \in \LANG(T)$, and by
  definition of $T$~\eqref{eq:LC:CONF}, it holds that $(\emptyset,
  \ddet[B](\initstateset, t)) \notin \LC(O,B)$, and the same holds for all
  prefixes of~$t$. Also $x \in
  \ddet[B](\initstateset, t)$, so there exists a trace $u \in \ACTstar$
  such that $x \ttrans[u\terminate]$, and for every prefix $r \prefix
  u\terminate$, it holds that $\ddet[O,B](\emptyset,
  \ddet[B](\initstateset, t), r) \notin \LC(O,B)$ (see
  \rem~\ref{rem:notLC}). By definition~\eqref{eq:LC:CONF}, it follows that
  $tu\terminate \in \LANG(T)$, and since $T$ is deterministic also $x_T
  \ttrans[u\terminate]$. Therefore, $B \sync T \ttrans[t] (x, x_T)
  \ttrans[u\terminate]$, i.e., $B \sync T$ is nonblocking.
\end{proof}

The result of \thm~\ref{thm:LC:CONF} shows how less conflicting pairs
generalise certain conflicts for the case when two automata are compared,
and in combination with the algorithm in \sect~\ref{sec:MC}, less
conflicting pairs lead to an alternative presentation of the
algorithm~\cite{Mal:10} to compute the set of certain conflicts.

\subsection{Testing the Conflict Preorder}
\label{sub:LC:confle}

Given the less conflicting pairs for two automata $A$ and~$B$, it is
possible to determine whether $A \confle B$. Automaton~$A$ is less
conflicting than~$B$ if every test~$T$ that is nonconflicting in
combination with~$B$ also is nonconflicting with~$A$. To check this
condition, it is enough to consider traces $B \sync T \ttrans[s]
(x_B,x_T)$, and check whether termination is also possible for every
state~$x_A$ of~$A$ such that $A \sync T \ttrans[s] (x_A,x_T)$. This amounts
to checking whether $(\{x_A\},X_B) \in \LC(A,B)$ when $A \ttrans[s]
x_A$ and $\ddet[B](\initstateset[B],s) = X_B$.

However, this condition does not apply to traces of certain conflicts. If
$s \in \CONF(B)$, then every test~$T$ that can execute~$s$ is in conflict
with~$B$. In this case, $A$ can still be less conflicting than~$B$, no
matter whether $A$ can or cannot execute the trace~$s$ and terminate
afterwards. This observation leads to the following result.

\begin{theorem}
  \label{thm:LC:confle}
  Let $A = \auttuple[A]$ and $B = \auttuple[B]$ be two automata.
  $A$ is less conflicting than $B$ if and only if
  for all $s \in \NCONF(B)$ and all $x_A \in Q_A$ such that $A \ttrans[s]
  x_A$ it holds that $(\{x_A\},X_B) \in \LC(A,B)$,
  where $\ddet[B](\initstateset[B], s) = X_B$.
\end{theorem}

\begin{proof}
  First assume that for all $s \in \NCONF(B)$ and all $x_A \in Q_A$ such
  that $A \ttrans[s] x_A$ it holds that $(\{x_A\},X_B) \in \LC(A,B)$, where
  $\ddet[B](\initstateset[B], s) = X_B$.
  Let $T = \auttuple[T]$ such that $B \sync T$ is nonblocking,
  and assume that $A \sync T \ttrans[s] (x_A,x_T)$.
  Since $B \sync T$ is nonblocking and $T \ttrans[s]$,
  it follows that $s \in \NCONF(B)$.
  Therefore by assumption $(\{x_A\}, X_B) \in \LC(A,B)$,
  so \refi{it:LC1:nonblocking1} or~\refi{it:LC1:blocking2} in
  \lemm~\ref{lem:LC1} must be true.
  However, \refi{it:LC1:blocking2} cannot hold,
  because for all $x_B \in X_B = \ddet[B](\initstateset,s)$
  it holds that $B \sync T \ttrans[s] (x_B,x_T)$,
  and since $B \sync T$ is nonblocking,
  there cannot exist any state $(y_B,y_T)$ such that 
  $(x_B,x_T) \ttrans (y_B,y_T)$ and
  $\LANGM(y_B,y_T) = \emptyset$.
  Thus, \refi{it:LC1:nonblocking1} must be true,
  and this means that $\LANGM(x_A,x_T) \neq \emptyset$.
  Since $T$ and~$s$ such that $A \sync T \ttrans[s] (x_A,x_T)$
  were chosen arbitrarily, it follows that $A \confle B$.

  Second assume that there exists $s \in \NCONF(B)$ and $x_A \in \Q[A]$ 
  such that $A \ttrans[s] x_A$ and $\XX = (\{x_A\}, X_B) \notin
  \LC(A,B)$, where $X_B = \ddet[B](\initstateset[B], s)$.
  Let $N_{B} = \detauttuple[N]$ be a deterministic recogniser of the
  language $\NCONF(B)$, and let $T_\XX = \detauttuple[T]$ be the
  deterministic automaton that exists according to \lemm~\ref{lem:LC2}.
  Since $s \in \NCONF(B)$, there exists a unique state $x_s \in \Q[N]$
  such that $N_{B} \trans[s] x_s$.
  Then construct the automaton
  \begin{equation}
    T = \langle \ACT, \Q[N] \dotcup \Q[T],
                \intrans_N \cup \intrans_T \cup \{(x_s,\tau,\initstate[T])\},
                \{\initstate[N]\}\rangle\ .
  \end{equation}

  Clearly, $A \sync T \ttrans[s] (x_A,x_s) \trans[\tau]
  (x_A\initphant,\initstate[T])$, and $\LANGM(x_A\initphant,\initstate[T])
  = \emptyset$ by \lemm~\ref{lem:LC2}~\refi{it:LC2:G1}. Thus, $A \sync T$
  is blocking.

  On the other hand, $B \sync T$ is nonblocking.
  To see this, consider $B \sync T \ttrans[t] (y_B,y_T)$.
  If $y_T \in \Q[N]$, then it follows from the fact that $B \sync
  N_B$ is nonblocking~\cite{MalStrRee:06} that there exists $u \in
  \ACTstar$ such that $(y_B,y_T) \ttrans[u\terminate]$.
  Otherwise $y_T \in \Q[T]$, which means that $t = su$ and
  $T \trans[s] x_s \trans[\tau] \initstate[T] \trans[u] y_T$.
  Also since $B \ttrans[t] y_B$, it follows that $y_B \in
  \ddet[B](\initstateset[B],t) = \ddet(\initstateset[B],su) =
  \ddet[B](\ddet[B](\initstateset[B],s),u) = \ddet[B](X_B,u)$, i.e.,
  there exists $x_B \in X_B$ such that $x_B \ttrans[u] y_B$. Thus
  $(x_B\initphant,\initstate[T]) \ttrans[u] (y_B,y_T)$, and by
  \lemm~\ref{lem:LC2}~\refi{it:LC2:G2}, it holds that
  $\LANGM(y_B,y_T) \neq \emptyset$.

  Thus, $A \sync T$ is blocking and $B \sync T$ is nonblocking,
  so $A \confle B$ cannot hold.
\end{proof}

\begin{example}
  Consider again automata $A_0$ and~$B_0$ in \fig~\ref{fig:blocking}.
  Recall that $\CONF(B_0) = \alpha\ACTstar$ from
  \examp~\ref{ex:certainconf}, so the only state in~$A_0$ that can be
  reached by a trace $s \notin \CONF(B_0)$ is~$a_0$. Therefore, it is
  enough to check the pair~$(\{a_0\}, \{b_0\})$ according to
  \thm~\ref{thm:LC:confle}, and it has been shown in \examp~\ref{ex:LC:0}
  that $(\{a_0\}, \{b_0\}) \in \LC^1(A_0,B_0)$.
  It follows that $A_0 \confle B_0$.
  This conclusion is made despite the fact that $(\{a_0\}, \{b_2\})
  \notin \LC(A_0,B_0)$, because $(\{a_0\}, \{b_2\})$ is only
  reachable by traces~$\alpha^n \in \CONF(B_0)$, $n\geq2$.
\end{example}

When using \thm~\ref{thm:LC:confle} to determine whether an automaton~$A$
is less conflicting than some blocking automaton~$B$, the set of certain
conflicts of~$B$ must be known first. This can be achieved using
\thm~\ref{thm:LC:CONF}, which makes it possible to classify state sets in
the subset construction of~$B$ as certain conflicts. If a state set $X_B
\subseteq \Q[B]$ is found to represent certain conflicts, i.e.,
$(\emptyset,X_B) \in \LC(O,B)$ according to \thm~\ref{thm:LC:CONF}, then
$(X_A, X_B) \in \LC(A,B)$ for every state set $X_A \subseteq\Q[A]$.
Successors reached only from such pairs are also certain conflicts of~$B$
and should not be considered when testing whether $A \confle B$ according
to \thm~\ref{thm:LC:confle}.

\def\statesize{\scriptsize}

\begin{figure}
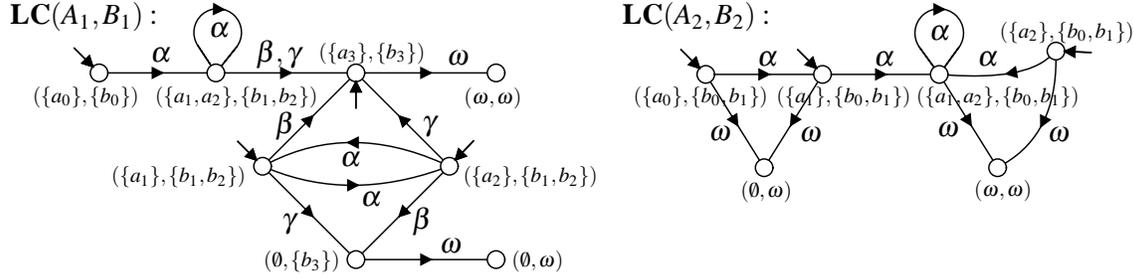

  \centering
  \tabcolsep0pt
  \psfrag{a}{$\alpha$}%
  \psfrag{o}{$\omega$}%
  \psfrag{ww}{\statesize$(\terminate,\terminate)$}%
  \psfrag{nullw}{\statesize$(\emptyset,\terminate)$}%
  \begin{tabular}{l@{\qquad}l}
    $\LC(A_1,B_1):$ &
    $\LC(A_2,B_2):$ \\
    \noalign{\vskip-3ex}
    \psfrag{b}{$\beta$}%
    \psfrag{c}{$\gamma$}%
    \psfrag{bc}{$\beta,\gamma$}%
    \psfrag{a0b0}{\statesize$(\{a_0\},\{b_0\})$}%
    \psfrag{a1a2b1b2}{\statesize$(\{a_1,a_2\},\{b_1,b_2\})$}%
    \psfrag{a3b3}{\statesize$(\{a_3\},\{b_3\})$}%
    \psfrag{a1b1b2}{\statesize$(\{a_1\},\{b_1,b_2\})$}%
    \psfrag{a2b1b2}{\statesize$(\{a_2\},\{b_1,b_2\})$}%
    \psfrag{nullb3}{\statesize$(\emptyset,\{b_3\})$}%
    \kern.4em
    \hangaut{LCA1B1} &
    \psfrag{a0b0b1}{\statesize$(\{a_0\},\{b_0,b_1\})$}%
    \psfrag{a1b0b1}{\statesize$(\{a_1\},\{b_0,b_1\})$}%
    \psfrag{a1a2b0b1}{\statesize$(\{a_1,a_2\},\{b_0,b_1\})$}%
    \psfrag{a2b0b1}{\statesize$(\{a_2\},\{b_0,b_1\})$}%
    \hangaut{LCA2B2}%
    \kern1.5em \\
  \end{tabular}
  \smallskip
  \caption{Less conflicting pairs for the automata pairs in
    \fig\ \ref{fig:notconfeq} and~\ref{fig:confeq}.}
  \label{fig:LC}
\end{figure}

\begin{example}
  Consider again automata $A_1$ and~$B_1$ in \fig~\ref{fig:notconfeq}.
  Composing $A_1$ with a deterministic version of~$B_1$ results in the
  following four pairs of states in~$A_1$ and sets of states in~$B_1$ that
  should be tested according to \thm~\ref{thm:LC:confle} to determine
  whether $A_1 \confle B_1$:
  \begin{equation}
    \label{eq:LCA1B1:init}
    (\{a_0\},\{b_0\})\quad
    (\{a_1\},\{b_1,b_2\})\quad
    (\{a_2\},\{b_1,b_2\})\quad
    (\{a_3\},\{b_3\}) \ .
  \end{equation}
  All four pairs need to be considered as $B_1$ is nonblocking and thus
  $\CONF(B_1) = \emptyset$.

  The graph to the left in \fig~\ref{fig:LC} shows these four pairs and their
  deterministic successors. The four pairs~\eqref{eq:LCA1B1:init} are
  marked as initial states, and the arrows in the graph represent the
  deterministic transition function. Although the deterministic transition
  function is defined for all state set pairs and events, arrows
  to~$(\emptyset,\emptyset)$ are suppressed for clarity of presentation.

  The following less conflicting pairs to compare $A_1$ to~$B_1$ are
  determined from the graph:
  \begin{align}
    (\terminate,\terminate) 
      & \in \LC^0(A_1,B_1) \ ; \\
    (\{a_0\}, \{b_0\}),\ (\{a_1, a_2\}, \{b_1, b_2\}),\ (\{a_3\}, \{b_3\})
      & \in \LC^1(A_1,B_1) \ .
  \end{align}
  For example, $(\{a_1, a_2\}, \{b_1, b_2\}) \in \LC^1(A_1,B_1)$, because
  all the ways to reach termination from state~$b_1$, i.e., all traces in
  $\LANGM(b_1) = \alpha^*\beta\terminate$ take the pair $(\{a_1, a_2\},
  \{b_1, b_2\})$ to $(\terminate,\terminate) \in LC^0(A_1,B_1)$. 
  No further pairs are found in~$\LC^2(A_1,B_1)$, so 
  $\LC(A_1,B_1)$ consists only of the pairs listed above. For example,
  $(\{a_1\}, \{b_1, b_2\}) \notin \LC^2(A_1,B_1)$, because the traces
  $\alpha\beta\terminate \in \LANGM(b_1)$ and $\gamma\terminate \in
  \LANGM(b_2)$ do not have any prefixes that reach a pair
  in~$\LC^1(A_1,B_1)$.

  As $(\{a_1\}, \{b_1, b_2\}) \notin \LC(A_1,B_1)$, it follows from
  \thm~\ref{thm:LC:confle} that $A_1$ is \emph{not} less conflicting
  than~$B_1$.
\end{example}

\begin{example}
  \label{ex:LC2}
  Consider again automata $A_2$ and~$B_2$ in \fig\ \ref{fig:confeq}.
  Again note that $\CONF(B_2) = \emptyset$. 
  By composing $A_2$ with a deterministic version of~$B_2$, it becomes
  clear that the only pairs that need to be tested to determine whether
  $A_2 \confle B_2$ according to \thm~\ref{thm:LC:confle} are
  $(\{a_0\},\{b_0,b_1\})$ reached after~$\varepsilon$,
  $(\{a_1\},\{b_0,b_1\})$ reached after~$\alpha^+$, and
  $(\{a_2\},\{b_0,b_1\})$ reached after~$\alpha\alpha^+$.

  The graph with these pairs and their deterministic successors is shown to
  the right in \fig~\ref{fig:LC}, with the three crucial pairs marked as
  initial. The following less conflicting pairs are discovered (see
  \examp~\ref{ex:LC:succ}):
  \begin{align}
    (\terminate,\terminate)   & \in \LC^0(A_2,B_2) \ ; \\
    \label{eq:LC1A2B2}
    (\{a_1\},\ \{b_0, b_1\}),\ (\{a_1,a_2\},\{b_0, b_1\}),\
      (\{a_2\}, \{b_0, b_1\}) & \in \LC^1(A_2,B_2) \ ; \\
    \label{eq:LC2A2B2}
    (\{a_0\}, \{b_0, b_1\})   & \in \LC^2(A_2,B_2) \ .
  \end{align}
  As the three crucial pairs are all in~$\LC(A_2,B_2)$, it follows from
  \thm~\ref{thm:LC:confle} that $A_2 \confle B_2$.
\end{example}

The result of \thm~\ref{thm:LC:confle} is related to the decision procedure
for fair testing~\cite{RenVog:07}. The fair testing decision procedure
starts by composing the automaton~$A$ with a determinised form of~$B$,
which gives rise to the same state set combinations that need to be
considered as in \thm~\ref{thm:LC:confle}. From this point on,
the two methods differ. The fair testing decision procedure annotates each
state of the synchronous product of~$A$ and the determinised form
of~$B$ with automata representing the
associated refusal trees, and searches for matching automata (or more
precisely, for matching \emph{productive subautomata}) within these
annotations. The method based on less conflicting pairs avoids some of the
resulting complexity by performing the complete decision on the flat state
space of the synchronous product of the determinised forms of $A$ and~$B$.

\section{Algorithm to Compute Less Conflicting Pairs}
\label{sec:MC}

This section proposes a method to effectively compute the less conflicting
pairs for two given finite-state automata $A$ and~$B$. This is done in a
nested iteration. Assuming that the set~$\LC^n(A,B)$ is already known, the
set~$\LC^{n+1}(A,B)$ is computed in a secondary iteration based on
\emph{more conflicting triples}.

\begin{definition}
  \label{def:MC}
  Let $A = \auttuple[A]$ and $B = \auttuple[B]$ be automata. The set
  $\MC^n(A,B) \subseteq \Qdet[A] \times \Qdet[B] \times \Q[B]\initphant$ of
  $n^{\mathrm{th}}$ level \emph{more conflicting triples} for $A$ and~$B$ is
  defined inductively as follows.
  \begin{align}
    \label{eq:MC:0}
    \MC^n_0(A,B) &= \{\, (\emptyset, \terminate, x_B) \mid
                             x_B \in \Q[B] \,\} \ ; \\
    \label{eq:MC:succ}
    \MC^n_{m+1}(A,B) &=
       \LongSet{11em}{$(X_A,X_B,x_B)$ $\mid$
                      $(X_A,X_B) \notin \LC^n(A,B)$ and $x_B \in X_B$ and
                      there exists $(Y_A, Y_B, y_B) \in \MC^n_m(A,B)$
                      and $\sigma \in \ACT$
                      such that $\ddet[A,B](X_A,X_B,\sigma) = (Y_A,Y_B)$
                      and $x_B \ttrans[\sigma] y_B$}; \\
    \label{eq:MC:all}
    \MC^n(A,B) &= \bigcup_{m\geq0} \MC^n_m(A,B) \ .
  \end{align}
\end{definition}

For a pair $(X_A,X_B)$ to be a less conflicting pair, according to
\defn~\ref{def:LC} there must be a state $x_B \in X_B$ such that every
trace that takes $x_B$ to termination in~$B$ has a prefix that leads to
another less conflicting pair. A triple $(X_A,X_B,x_B)$ is considered
``more conflicting'' if $(X_A,X_B)$ is not yet known to be a less
conflicting pair, and the state $x_B \in X_B$ cannot be used to confirm the
above property. Therefore, \lemm~\ref{lem:MC} shows that a triple
$(X_A,X_B,x_B)$ is $n^{\mathrm{th}}$-level ``more conflicting'' if and only
if the state~$x_B \in X_B$ can reach termination without passing through a
pair in~$\LC^n$.

If $(X_A,X_B,x_B)$ is ``more conflicting'' for all $x_B \in X_B$, then the
pair $(X_A,X_B)$ cannot be a less conflicting pair.
Otherwise, if there exists at least one state $x_B \in X_B$ such that
$(X_A,X_B,x_B)$ is not ``more conflicting'', then $(X_A,X_B)$ is added to
set of less conflicting pairs in the next iteration. \Thm~\ref{thm:MC}
below confirms the correctness of this approach.

\begin{lemma}
  \label{lem:MC}
  Let $A = \auttuple[A]$ and $B = \auttuple[B]$ be automata, let $n \in
  \NAT$ and $(X_A, X_B, x_B) \in \Qdet[A] \times \Qdet[B] \times
  \Q[B]\initphant$. The following statements are equivalent.
  \begin{enumerate}
  \item \label{it:MC:MCn}
    $(X_A, X_B, x_B) \in \MC^n(A,B)$;
  \item \label{it:MC:trace}
    There exists a trace $s \in \ACTstar\terminate \cup \{\varepsilon\}$
    such that $\ddet[A,B](X_A,X_B,s) = (\emptyset,\terminate)$ and $x_B
    \ttrans[s]$, and $\ddet[A,B](X_A,X_B,r) \notin \LC^n(A,B)$ for all
    prefixes $r \prefix s$.
  \end{enumerate}
\end{lemma}

\begin{proof}
  First let $(X_A, X_B, x_B) \in \MC^n(A,B)$, i.e.,
  $(X_A, X_B, x_B) \in \MC^n_m(A,B)$ for some $m \in \NAT$.
  It is shown by induction on~$m$ that \refi{it:MC:trace} holds.

  In the base case, $m = 0$, and by definition $(X_A,X_B,x_B) \in
  \MC^n_0(A,B)$ means that $(X_A,X_B) = (\emptyset,\terminate)$. Then
  consider $s = \varepsilon$, and note $\ddet[A,B](X_A,X_B,\varepsilon) =
  (X_A,X_B) = (\emptyset,\terminate)$ and $x_B \ttrans[\varepsilon]$.
  Clearly $r \prefix \varepsilon$ implies $r = \varepsilon$, and
  $\ddet[A,B](X_A,X_B,\varepsilon) = (\emptyset,\terminate)
  \notin \LC(A,B) \supseteq \LC^n(A,B)$ by \lemm~\ref{lem:LC1}.

  Now consider $(X_A, X_B, x_B) \in \MC^n_{m+1}(A,B)$.
  It follows from \defn~\ref{def:MC} that $(X_A, X_B) \notin LC^n(A,B)$
  and $x_B \in X_B$, and there exists $(Y_A, Y_B, y_B) \in \MC^n_m(A,B)$
  and $\sigma \in \ACT$ such that $\ddet[A,B](X_A,X_B,\sigma) = (Y_A,Y_B)$ and
  $x_B \ttrans[\sigma] y_B$.
  By inductive assumption, there exists a trace $s \in \ACTstar\terminate
  \cup \{\varepsilon\}$ such that $\ddet[A,B](Y_A,Y_B,s) =
  (\emptyset,\terminate)$ and $y_B \ttrans[s]$, and for
  all $r \prefix s$ it holds that $\ddet[A,B](Y_A,Y_B,r) \notin \LC^n(A,B)$.
  Then $\ddet[A,B](X_A,X_B,\sigma s) = \ddet[A,B](Y_A,Y_B, s) =
  (\emptyset,\terminate)$ and $x_B \ttrans[\sigma] y_B \ttrans[s]$, and for
  all $r \prefix \sigma s$ it holds that $\ddet[A,B](X_A,X_B,r) \notin
  \LC^n(A,B)$.

  Conversely, let $s \in \ACTstar\terminate \cup \{\varepsilon\}$ such that
  \refi{it:MC:trace} holds. This means that $\ddet[A,B](X_A,X_B,s) =
  (\emptyset,\terminate)$ and $x_B \ttrans[s]$, and $\ddet[A,B](X_A,X_B,r)
  \notin \LC^n(A,B)$ for all $r \prefix s$. It is shown by induction on $m
  = |s|$ that $(X_A,X_B,x_B) \in \MC^n_m(A,B)$.

  In the base case, $m = 0$ and $s = \varepsilon$, it holds by definition that
  $(X_A,X_B) = \ddet[A,B](X_A,X_B,\varepsilon) = (\emptyset,\terminate) \in
  \MC^n_0(A,B)$.

  Now let $s = \sigma t$ such that $|t|=m$, and $\ddet[A,B](X_A,X_B,s) =
  (\emptyset,\terminate)$ and $x_B \ttrans[s]$, and 
  $\ddet[A,B](X_A,X_B,r) \notin \LC^n(A,B)$ for all prefixes $r \prefix s$.
  Write $\ddet[A,B](X_A,X_B,\sigma) = (Y_A,Y_B)$ and $x_B \ttrans[\sigma] y_B
  \ttrans[t]$.
  Then $y_B \ttrans[t]$ and $\ddet[A,B](Y_A,Y_B,t) =
  \ddet[A,B](X_A,X_B,\sigma t) = \ddet[A,B](X_A,X_B,s) =
  (\emptyset,\terminate)$ and $\ddet[A,B](Y_A,Y_B,r) \notin \LC^n(A,B)$
  for all $r \prefix t$.
  Then $(Y_A,Y_B,y_B) \in \MC^n_m(A,B)$ by inductive assumption, and by
  \defn~\ref{def:MC} it follows that $(X_A,X_B,x_B) \in \MC^n_{m+1}(A,B)$.
\end{proof}

\begin{theorem}
  \label{thm:MC}
  Let $A = \auttuple[A]$ and $B = \auttuple[B]$ be automata, and let
  $n \in \NAT$. Then
  \begin{equation}
    \LC^{n+1}(A,B) = \{\, (X_A,X_B) \in \Qdet[A] \times \Qdet[B] \mid
                          (X_A,X_B,x_B) \notin \MC^n(A,B)\
                          \mbox{for some}\ x_B \in X_B \,\}\ .
  \end{equation}
\end{theorem}

\begin{proof}
  Let $(X_A, X_B) \in LC^{n+1}(A,B)$.
  Then by \defn~\ref{def:LC}, there exists $x_B \in X_B$ such that for all
  $t \in \ACTstar$ such that $x_B \ttrans[t\terminate]$, there exists $r
  \prefix t\terminate$ such that $\ddet[A,B](X_A, X_B, r) \in \LC^i(A,B)$
  for some $i \leq n$.
  Equivalently, this means that there does not exist a trace $t \in
  \ACTstar$ such that $x_B \ttrans[t\terminate]$ and for all prefixes $r
  \prefix t\terminate$ it holds that $\ddet[A,B](X_A, X_B, r) \notin
  \LC^n(A,B)$.
  Then $(X_A, X_B, x_B) \notin \MC^n(A,B)$ because otherwise such a
  trace would exist by \lemm~\ref{lem:MC}.

  Conversely, let $x_B \in X_B$ such that $(X_A, X_B, x_B) \notin
  \MC^n(A,B)$.
  To check the condition in \defn~\ref{def:LC}~\eqref{eq:LC:succ},
  consider $t \in \ACTstar$ such that $x_B \ttrans[t\terminate]$.
  Then clearly $\ddet[B](X_B,t\terminate) = \terminate$.
  By \defn~\ref{def:ddet}, it holds that either
  $\ddet[A](X_A,t\terminate) = \terminate$ or
  $\ddet[A](X_A,t\terminate) = \emptyset$.
  If $\ddet[A](X_A,t\terminate) = \terminate$,
  then $\ddet[A,B](X_A,X_B,t\terminate) = (\terminate,\terminate) \in
  \LC^0(A,B)$.
  Otherwise $\ddet[A](X_A,t\terminate) = \emptyset$
  and thus $\ddet[A,B](X_A,X_B,t\terminate) = (\emptyset,\terminate)$,
  and by \lemm~\ref{lem:MC} there must exist $r \prefix t\terminate$
  such that $\ddet[A,B](X_A,X_B,r) \in \LC^n(A,B)$
  as otherwise $(X_A,X_B,x_B) \in \MC^n(A,B)$.
  In both cases, $\ddet[A,B](X_A,X_B,r) \in \LC^i(A,B)$ for some $r \prefix
  t\terminate$ and $i \leq n$.
  Since $t \in \ACTstar$ with $x_B \ttrans[t\terminate]$ was chosen
  arbitrarily, it follows from \defn~\ref{def:LC}~\eqref{eq:LC:succ} that
  $(X_A,X_B) \in \LC^{n+1}(A,B)$.
\end{proof}

\begin{example}
\begin{figure}
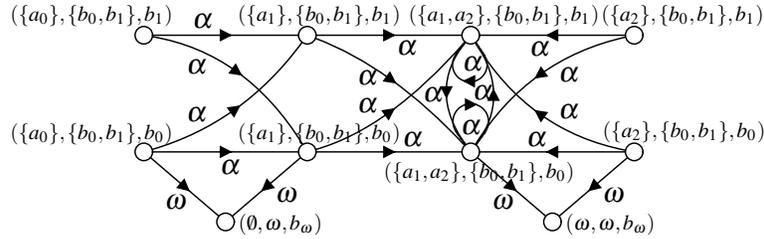

  \psfrag{a}{$\alpha$}
  \psfrag{o}{$\terminate$}
  \psfrag{ww}{\statesize$(\terminate,\terminate,b_\terminate)$}%
  \psfrag{nullw}{\statesize$(\emptyset,\terminate,b_\terminate)$}%
  \psfrag{a0b0}{\statesize$(\{a_0\},\{b_0,b_1\},b_0)$}%
  \psfrag{a1b0}{\statesize$(\{a_1\},\{b_0,b_1\},b_0)$}%
  \psfrag{a1a2b0}{\statesize$(\{a_1,a_2\},\{b_0,b_1\},b_0)$}%
  \psfrag{a2b0}{\statesize$(\{a_2\},\{b_0,b_1\},b_0)$}%
  \psfrag{a0b1}{\statesize$(\{a_0\},\{b_0,b_1\},b_1)$}%
  \psfrag{a1b1}{\statesize$(\{a_1\},\{b_0,b_1\},b_1)$}%
  \psfrag{a1a2b1}{\statesize$(\{a_1,a_2\},\{b_0,b_1\},b_1)$}%
  \psfrag{a2b1}{\statesize$(\{a_2\},\{b_0,b_1\},b_1)$}%
  \centerline{\autgraphics{MCA2B2}\kern4em}
  \caption{Calculating more conflicting triples for automata $A_2$
    and~$B_2$ in \fig~\ref{fig:confeq}.}
  \label{fig:MC}
\end{figure}
  \Fig~\ref{fig:MC} shows a graph representing the more conflicting triples
  to check whether $A_2 \confle B_2$ in \fig~\ref{fig:confeq}. The arrows
  in the graph represent the deterministic transition function in
  combination with the transition relation of~$B_2$. An arrow
  $(X_A,X_B,x_B) \trans[\sigma] (Y_A,Y_B,y_B)$ indicates that
  $\ddet[A_2,B_2](X_A,X_B,\sigma) = (Y_A,Y_B)$ and $x_B \ttrans[\sigma]
  y_B$.

  In the first iteration to compute~$\MC^0(A_2,B_2)$, first the triple
  $(\emptyset, \terminate, b_\terminate)$ is added to $\MC^0_0(A_2,B_2)$.
  Next, the 
  triples $(\{a_0\}, \{b_0, b_1\}, b_0)$ and $(\{a_1\}, \{b_0, b_1\}, b_0)$
  are added to $\MC^0_1(A_2,B_2)$ as they can immediately reach $(\emptyset,
  \terminate, b_\terminate)$. Finally, $(\{a_0\}, \{b_0, b_1\}, b_1)$ is
  also added to $\MC^0_2(A_2,B_2)$ as it reaches $(\{a_1\}, \{b_0, b_1\},
  b_0) \in \MC^0_1(A_2,B_2)$. No further triples are found to be in
  $\MC^0_3(A_2,B_2)$. Therefore, $(\{a_1\}, \{b_0, b_1\}, b_1) \notin
  \MC^0(A_2,B_2)$, so it follows from \thm~\ref{thm:MC} that $(\{a_1\},
  \{b_0, b_1\}) \in \LC^1(A_2,B_2)$, and likewise $(\{a_1,a_2\}, \{b_0,
  b_1\})\bbcom (\{a_2\}, \{b_0, b_1\}) \in \LC^1(A_2,B_2)$.

  In the next iteration to compute~$\MC^1(A_2,B_2)$, note that $(\{a_1\},
  \{b_0, b_1\}, b_0) \notin \MC^1_1(A_2,B_2)$ because $(\{a_1\}, \{b_0,
  b_1\}) \in \LC^1(A_2,B_2)$. Still, $(\{a_0\}, \{b_0, b_1\}, b_0) \in
  \MC^1_1(A_2,B_2)$ because of the transition to $(\emptyset, \terminate,
  b_\terminate) \in \MC^1_0(A_2,B_2)$, but $(\{a_0\}, \{b_0, b_1\}, b_1)
  \notin \MC^1_2(A_2,B_2)$ because now $(\{a_1\}, \{b_0, b_1\}, b_0) \notin
  \MC^1_1(A_2,B_2)$. Accordingly, the pair $(\{a_0\}, \{b_0, b_1\})$ is
  added to $\LC^2(A_2,B_2)$.

  In a final iteration to compute~$\MC^2(A_2,B_2)$, only one more
  conflicting triple is found, $(\emptyset, \terminate, b_\terminate) \in
  \MC^2_0(A_2,B_2)$. No further pairs are added in~$\LC^3(A_2,B_2)$.
  At this point, the iteration terminates, having found exactly the four
  less conflicting pairs given in \examp~\ref{ex:LC2}, \eqref{eq:LC1A2B2}
  and~\eqref{eq:LC2A2B2}.
\end{example}

To determine whether an automaton~$A$ is less conflicting than automaton~$B$,
it is first needed to determine the set of certain conflicts of~$B$, and
then to find all the state-set pairs for $A$ and~$B$ that are reachable from a
pair like $(\{x_A\}, X_B)$ associated with some trace that is not a certain
conflict of~$B$. The more conflicting triples can be constructed as they
are discovered during the backwards search from the terminal states.

The complexity of each iteration of the more conflicting triples computation
is determined by the number of arrows in the graph, which is bounded by
$|\ACT| \cdot |Q_B|^2 \cdot 2^{|Q_A|} \cdot 2^{|Q_B|}$, because the
powerset transitions are deterministic, which is not the case for the
transitions of~$B$. Each iteration except the last adds at least one less
conflicting pair, so the number of iterations is bounded by $2^{|Q_A|}
\cdot 2^{|Q_B|}$. The complexity of this loop dominates all other tasks of
the computation. Therefore, the worst-case time complexity to determine
whether $A \confle B$ using less conflicting pairs is
\begin{equation}
  \label{eq:complexity}
  O(|\ACT| \cdot |Q_B|^2 \cdot 4^{|Q_A|} \cdot 4^{|Q_B|}) \; = \;
  O(|\ACT| \cdot |Q_B|^2 \cdot 2^{2|Q_A| + 2|Q_B|}) \ .
\end{equation}
This shows that the conflict preorder can be tested in linear exponential
time, as it is the case for the fair testing preorder. Yet, the
complexity~\eqref{eq:complexity} is better than the time complexity of the
decision procedure for fair testing, which is $O(|Q_A| \cdot |Q_B| \cdot
2^{3|Q_A| + 5|Q_B|})$~\cite{RenVog:07}.


\section{Conclusions}
\label{sec:conclusions}

Less conflicting pairs provide a concrete state-based means to characterise
the extent by which one process is or is not less conflicting than another.
The characterisation generalises and includes previous results about
certain conflicts, and it gives rise to a direct way to test the conflict
preorder and the related fair testing preorder by inspecting sets of
reachable states. Based on the characterisation, an effective algorithm is
presented to test whether a finite-state automaton is less conflicting than
another. The algorithm, while still linear exponential, has better time
complexity than the previously known decision procedure for fair testing.

In the future, the authors would like to apply the theoretic results of
this paper to compute abstractions and improve the performance of
compositional model checking algorithms. The more thorough understanding of
the conflict preorder will make it possible to better simplify processes
with respect to conflict equivalence and other related liveness properties.


\bibliography{%
  ../references/IEEEabrv,../references/malik_abrv,../references/malik}

\begin{thebibliography}{10}
\providecommand{\bibitemdeclare}[2]{}
\providecommand{\urlprefix}{Available at }
\providecommand{\url}[1]{\texttt{#1}}
\providecommand{\href}[2]{\texttt{#2}}
\providecommand{\urlalt}[2]{\href{#1}{#2}}
\providecommand{\doi}[1]{doi:\urlalt{http://dx.doi.org/#1}{#1}}
\providecommand{\bibinfo}[2]{#2}

\bibitemdeclare{book}{BaiKat:08}
\bibitem{BaiKat:08}
\bibinfo{author}{Christel Baier} \& \bibinfo{author}{Joost-Pieter Katoen}
  (\bibinfo{year}{2008}): \emph{\bibinfo{title}{Principles of Model Checking}}.
\newblock \bibinfo{publisher}{MIT Press}.

\bibitemdeclare{inproceedings}{BriRenVog:95}
\bibitem{BriRenVog:95}
\bibinfo{author}{Ed~Brinksma}, \bibinfo{author}{Arend Rensink} \&
  \bibinfo{author}{Walter Vogler} (\bibinfo{year}{1995}):
  \emph{\bibinfo{title}{Fair Testing}}.
\newblock In \bibinfo{editor}{Insup Lee} \& \bibinfo{editor}{Scott~A. Smolka},
  editors: {\sl \bibinfo{booktitle}{Proc. 6th Int. Conf. Concurrency Theory,
  {CONCUR}\,'95}}, {\sl \bibinfo{series}{LNCS}} \bibinfo{volume}{962},
  \bibinfo{publisher}{Springer}, \bibinfo{address}{Philadelphia, PA, USA}, pp.
  \bibinfo{pages}{313--327}.

\bibitemdeclare{book}{CasLaf:99}
\bibitem{CasLaf:99}
\bibinfo{author}{C.~G. Cassandras} \& \bibinfo{author}{S.~Lafortune}
  (\bibinfo{year}{1999}): \emph{\bibinfo{title}{Introduction to Discrete Event
  Systems}}.
\newblock \bibinfo{publisher}{Kluwer}.

\bibitemdeclare{article}{DeNHen:84}
\bibitem{DeNHen:84}
\bibinfo{author}{R.~{De Nicola}} \& \bibinfo{author}{M.~C.~B. Hennessy}
  (\bibinfo{year}{1984}): \emph{\bibinfo{title}{Testing Equivalences for
  Processes}}.
\newblock {\sl \bibinfo{journal}{Theoretical Comput. Sci.}}
  \bibinfo{volume}{34}(\bibinfo{number}{1--2}), pp. \bibinfo{pages}{83--133},
  \doi{10.1016/0304-3975(84)90113-0}.

\bibitemdeclare{article}{FloMal:09}
\bibitem{FloMal:09}
\bibinfo{author}{Hugo Flordal} \& \bibinfo{author}{Robi Malik}
  (\bibinfo{year}{2009}): \emph{\bibinfo{title}{Compositional Verification in
  Supervisory Control}}.
\newblock {\sl \bibinfo{journal}{SIAM J. Control and Optimization}}
  \bibinfo{volume}{48}(\bibinfo{number}{3}), pp. \bibinfo{pages}{1914--1938},
  \doi{10.1137/070695526}.

\bibitemdeclare{incollection}{vGl:01}
\bibitem{vGl:01}
\bibinfo{author}{R.~J. van Glabbeek} (\bibinfo{year}{2001}):
  \emph{\bibinfo{title}{The Linear Time~--- Branching Time Spectrum~{I}: The
  Semantics of Concrete, Sequential Processes}}.
\newblock In \bibinfo{editor}{J.~A. Bergstra}, \bibinfo{editor}{A.~Ponse} \&
  \bibinfo{editor}{S.~A. Smolka}, editors: {\sl \bibinfo{booktitle}{Handbook of
  Process Algebra}}, \bibinfo{publisher}{Elsevier}, pp. \bibinfo{pages}{3--99}.

\bibitemdeclare{book}{Hoa:85}
\bibitem{Hoa:85}
\bibinfo{author}{C.~A.~R. Hoare} (\bibinfo{year}{1985}):
  \emph{\bibinfo{title}{Communicating Sequential Processes}}.
\newblock \bibinfo{publisher}{Prentice-Hall}.

\bibitemdeclare{book}{HopMotUll:01}
\bibitem{HopMotUll:01}
\bibinfo{author}{John~E. Hopcroft}, \bibinfo{author}{Rajeev Motwani} \&
  \bibinfo{author}{Jeffrey~D. Ullman} (\bibinfo{year}{2001}):
  \emph{\bibinfo{title}{Introduction to Automata Theory, Languages, and
  Computation}}.
\newblock \bibinfo{publisher}{Addison-Wesley}.

\bibitemdeclare{inproceedings}{Mal:04}
\bibitem{Mal:04}
\bibinfo{author}{Robi Malik} (\bibinfo{year}{2004}): \emph{\bibinfo{title}{On
  the Set of Certain Conflicts of a Given Language}}.
\newblock In: {\sl \bibinfo{booktitle}{Proc. 7th Int. Workshop on Discrete
  Event Systems, {WODES}\,\,'04}}, \bibinfo{address}{Reims, France}, pp.
  \bibinfo{pages}{277--282}.

\bibitemdeclare{misc}{Mal:10}
\bibitem{Mal:10}
\bibinfo{author}{Robi Malik} (\bibinfo{year}{2010}): \emph{\bibinfo{title}{The
  Language of Certain Conflicts of a Nondeterministic Process}}.
\newblock \bibinfo{howpublished}{Working Paper 05/2010, Dept. of Computer
  Science, University of Waikato, Hamilton, New Zealand}.

\bibitemdeclare{article}{MalStrRee:06}
\bibitem{MalStrRee:06}
\bibinfo{author}{Robi Malik}, \bibinfo{author}{David Streader} \&
  \bibinfo{author}{Steve Reeves} (\bibinfo{year}{2006}):
  \emph{\bibinfo{title}{Conflicts and Fair Testing}}.
\newblock {\sl \bibinfo{journal}{Int. J. Found. Comput. Sci.}}
  \bibinfo{volume}{17}(\bibinfo{number}{4}), pp. \bibinfo{pages}{797--813}.

\bibitemdeclare{book}{Mil:89}
\bibitem{Mil:89}
\bibinfo{author}{Robin Milner} (\bibinfo{year}{1989}):
  \emph{\bibinfo{title}{Communication and concurrency}}.
\newblock \bibinfo{series}{Series in Computer Science},
  \bibinfo{publisher}{Prentice-Hall}.

\bibitemdeclare{inproceedings}{NatCle:95}
\bibitem{NatCle:95}
\bibinfo{author}{V.~Natarajan} \& \bibinfo{author}{Rance Cleaveland}
  (\bibinfo{year}{1995}): \emph{\bibinfo{title}{Divergence and Fair Testing}}.
\newblock In: {\sl \bibinfo{booktitle}{Proc. 22nd Int. Colloquium on Automata,
  Languages, and Programming, {ICALP}\,'95}}, pp. \bibinfo{pages}{648--659}.

\bibitemdeclare{article}{RamWon:89}
\bibitem{RamWon:89}
\bibinfo{author}{Peter J.~G. Ramadge} \& \bibinfo{author}{W.~Murray Wonham}
  (\bibinfo{year}{1989}): \emph{\bibinfo{title}{The Control of Discrete Event
  Systems}}.
\newblock {\sl \bibinfo{journal}{Proc. {IEEE}}}
  \bibinfo{volume}{77}(\bibinfo{number}{1}), pp. \bibinfo{pages}{81--98}.

\bibitemdeclare{article}{RenVog:07}
\bibitem{RenVog:07}
\bibinfo{author}{Arend Rensink} \& \bibinfo{author}{Walter Vogler}
  (\bibinfo{year}{2007}): \emph{\bibinfo{title}{Fair testing}}.
\newblock {\sl \bibinfo{journal}{Information and Computation}}
  \bibinfo{volume}{205}(\bibinfo{number}{2}), pp. \bibinfo{pages}{125--198},
  \doi{10.1016/j.ic.2006.06.002}.

\bibitemdeclare{article}{SuSchRooHof:10}
\bibitem{SuSchRooHof:10}
\bibinfo{author}{Rong Su}, \bibinfo{author}{Jan~H. van Schuppen},
  \bibinfo{author}{Jacobus~E. Rooda} \& \bibinfo{author}{Albert~T. Hofkamp}
  (\bibinfo{year}{2010}): \emph{\bibinfo{title}{Nonconflict check by using
  sequential automaton abstractions based on weak observation equivalence}}.
\newblock {\sl \bibinfo{journal}{Automatica}}
  \bibinfo{volume}{46}(\bibinfo{number}{6}), pp. \bibinfo{pages}{968--978},
  \doi{10.1016/j.automatica.2010.02.025}.

\bibitemdeclare{inproceedings}{WarMal:10}
\bibitem{WarMal:10}
\bibinfo{author}{Simon Ware} \& \bibinfo{author}{Robi Malik}
  (\bibinfo{year}{2010}): \emph{\bibinfo{title}{Compositional Nonblocking
  Verification Using Annotated Automata}}.
\newblock In: {\sl \bibinfo{booktitle}{Proc. 10th Int. Workshop on Discrete
  Event Systems, {WODES}\,'10}}, \bibinfo{address}{Berlin, Germany}, pp.
  \bibinfo{pages}{374--379}.

\bibitemdeclare{article}{WonThiMalHoa:00}
\bibitem{WonThiMalHoa:00}
\bibinfo{author}{K.~C. Wong}, \bibinfo{author}{J.~G. Thistle},
  \bibinfo{author}{R.~P. Malhame} \& \bibinfo{author}{H.-H. Hoang}
  (\bibinfo{year}{2000}): \emph{\bibinfo{title}{Supervisory Control of
  Distributed Systems: Conflict Resolution}}.
\newblock {\sl \bibinfo{journal}{Discrete Event Dyn. Syst.}}
  \bibinfo{volume}{10}, pp. \bibinfo{pages}{131--186}.

\end{thebibliography}

\end{document}